\documentclass[12pt]{article}
\usepackage{verbatim}
\usepackage{amsfonts}
\usepackage{graphics}
\usepackage{amsmath}
\usepackage{times}
\usepackage{appendix}
\usepackage{color}
\usepackage{soul}
\usepackage{mathtools}
\usepackage{enumerate}
\usepackage{fancyhdr,latexsym,amsmath,amsfonts,amssymb,amsbsy,amsthm,url}
\usepackage[margin=0.5in,footskip=0.25in]{geometry}
\usepackage{graphics,graphicx,epsfig}
\usepackage{bbm}
\usepackage{breqn}
\usepackage{caption,subcaption,float,subfloat,pgfplots}
\usepackage{pdflscape}
\usepackage{hyperref}
\usepackage{tikz}
\usepackage[ruled,vlined]{algorithm2e}

\newtheorem{theorem}{Theorem}

\newtheorem{corollary}[theorem]{Corollary}

\newtheorem{assumption}{Assumption}
\newtheorem{model}{Model}[]

\usepackage[english]{babel}
\makeatletter

\renewcommand{\section}{
	\@startsection
	{section}
	{1}
	{0pt}
	{1.1\baselineskip}
	{0.2\baselineskip}
	{\sc \centering}
}

\renewcommand{\subsection}{
	\@startsection
	{subsection}
	{1}
	{0pt}
	{1.1\baselineskip}
	{0.2\baselineskip}
	{\sc \centering}
}

\renewcommand{\subsubsection}{
	\@startsection
	{subsubsection}
	{1}
	{0pt}
	{1.1\baselineskip}
	{0.2\baselineskip}
	{\sc \centering}
}

\makeatother

\usepackage[flushleft]{threeparttable}
\usepackage{rotating,booktabs,multirow}
\usepackage{colortbl}
\usepackage{makecell,cellspace,caption}

\begin{document}
	
\title{\large\sc Dynamic loan portfolio management in a three time step model}
\normalsize
\author{
\sc{Deb Narayan Barik} \thanks{Department of Mathematics, Indian Institute of Technology Guwahati, Guwahati-781039, India, e-mail: d.narayan@iitg.ac.in}
\and 
\sc{Siddhartha P. Chakrabarty} \thanks{Department of Mathematics, Indian Institute of Technology Guwahati, Guwahati-781039, India, e-mail: pratim@iitg.ac.in, Phone: +91-361-2582606}}

\date{}
\maketitle
\begin{abstract}

This paper studies the bank dynamic decision problem in the intermediate time step for a discrete-time setup. We have considered a three-time-step model. Initially, the banks raise money through debt and equity and invest in different types of loans. It liquidates its assets and raises new funds at the intermediate-time step to meet the short-term debt holder’s claim. Further, it has to meet specific capital requirements given by the regulators. In this work, we have theoretically studied the effect of raising new equity and debt. We show that in some cases, raising equity and debt may increase the return on equity, and in some cases, it may decrease the return on equity. We have discussed several cases and given a bound on the capital that can be raised. We have added an equity holder's constraint, which ensures the return on equity and desists the bank from defaulting at the final time point.



\end{abstract}

\section{Introduction}
\label{Sec_Three_Introduction}
 
Financial regulators have introduced the capital requirement to discourage banks or other financial institutions from taking excess leverage and risk. Basel I and Basel II \cite{Basel1, Basel2} have introduced risk-based capital measures. However several incompleteness in these definitions has broadened the path to impose more powerful regulation on the banks and financial institutions. Basel III \cite{Basel3Leverage}, introduced Leverage Ratio (LR), a non-risk-based risk metric, that serves as credible supplementary metric to the risk-based requirement. The definition of LR is given by
\[\text{LR}:=\frac{\text{Tier 1 capital}}{\text{Total exposure measure}} ,\] 
with this ratio being a counter-cyclical capital measure. BIS has identified that pro-cyclical measures are not sufficient to serve the purpose of stabilizing banks \cite{Basel3Leverage}, and therefore they have introduced a counter cyclical measure. A countercyclical measure is more effective at preventing financial crisis, by reducing systematic risk and credit bubbles. In \cite{Blum2008}, the author has shown that the risk-based capital requirement along with the Leverage Ratio serves better as a capital requirement criterion. D'Hulster \cite{D2009} has discussed about bank leverage and various aspects of the Leverage Ratio. Philipp M Hildebrand \cite{Hildebrand2008} has claimed that implementing risk-based capital requirements with a Leverage Ratio lowers the leverage of the bank. This eventually decreases the chance of default of the bank. In this work we have considered the loan portfolio which fulfill both the criteria (risk based capital requirement and Leverage Ratio).

It is well known that the bank owner is protected by Limited Liability which has provision that an bank owner is not liable to it's personal property, in case of bankruptcy of the bank. In \cite{Carney1998}, the author has discussed the history of Limited Liability. This extra protection causes moral hazard problems as discussed by \cite{Sinn2001, Biais2010} and many other authors. Acosta-Smith et.al. \cite{Acosta20} have mentioned the bank's decision problem with Limited Liability along with latest capital requirement given by Basel III. We have shown, in the article \cite{Barik2022}, that inclusion of Limited Liability in the decision model can decreases leveraged risk. With limited liability protection, in order to reach a target, the required amount of risk is less than that of the case without Limited Liability protection. 
 
Banks' problems in the intermediate time step are crucial for business cycles. It involves raising money, liquidating assets, meeting liabilities etc. 
Beside all these cash flows, the bank has to meet the regulatory requirements. Therefore modelling this phenomenon is an important and interesting research topic. The authors in \cite{Behn2019} and \cite{Cohen2013}, have discussed dynamic bank behaviour to comply with the regulatory requirements. Behn et al. \cite{Behn2019} have suggested four methods to meet the capital ratio requirements. In this study, we have analyzed these four. We have constructed a model which includes these four strategies, so that the bank can choose among these four methods to optimize their goals.

The outline of the paper is as follows. We begin with an detailed presentation of the models in Section \ref{Sec_Three_Problem_and_Model_Formulation}. Then, in Section \ref{Sec_Three_Theoretical Results} we present all the theoretical results, which are then illustrated by an example in section \ref{Sec_Three_An_Example}. Finally, Section \ref{Sec_Three_Conclusions} summarizes the main takeaways and concluding remarks of the work carried out.

\section{Problem and Model Formulation}
\label{Sec_Three_Problem_and_Model_Formulation}

For the problem setup, the banks are assumed to follow the firm based model, wherein they raise capital through debt and equity. The capital raised is assumed to be invested by the bank in three type of (loan) assets, namely, one safe asset and two risky assets. For our study, we have considered the contagion risk, and have adopted the payoff structure to be the same as in \cite{Acosta20}. The authors in \cite{Acosta20} have discussed about two type of assets, with the condition that when the safer asset defaults, then the risky assets also default. A dynamic model for bank behaviour is discussed in \cite{Behn2019}, with the incorporation of various regulatory provisions. In particular, the model in \cite{Behn2019} illustrates four different strategies, in order to adjust the capital ratio, so as to meet the regulatory requirements.

We begin by presenting a schematic diagram of the proposed model, in Figure \ref{Fig:Three_One}. We denote the safe asset as $L_{0}$ and the risky assets by $L_{1}$ and $L_{2}$ (with $L_{1}$ being less risky than $L_{2}$). Further the suffix $D$ indicates a default event, that is, $\{L_{i}D\}_{i=1,2}$ denotes the default of the asset $L_{i}$. Further, the description of the various parameters, to be used in the model formulation, are presented in Table \ref{Tab:2.1}.

\begin{figure}[ht]
\begin{center}	
\begin{tikzpicture}[sibling distance=15mm]
\node[rectangle,draw,xshift=0cm,yshift=0cm] (00) {$L_{0},L_{1},L_{2}$};
		
\node[rectangle,draw,xshift=5cm,yshift=2cm] (11) {$L_{0},L_{1},L_{2}$};
\node[rectangle,draw,xshift=5cm,yshift=-0.75cm] (12) {$L_{0},L_{1},L_{2}{\color{red}{D}}$};
\node[rectangle,draw,xshift=5cm,yshift=-2cm] (13) {$L_{0},L_{1}{\color{red}{D}},L_{2}{\color{red}{D}}$};
		
\node[rectangle,draw,xshift=10cm,yshift=3cm] (21) {$L_{0},L_{1},L_{2}$};
\node[rectangle,draw,xshift=10cm,yshift=2cm] (22) {$L_{0},L_{1},L_{2}{\color{red}{D}}$};
\node[rectangle,draw,xshift=10cm,yshift=1cm] (23) {$L_{0},L_{1}{\color{red}{D}},L_{2}{\color{red}{D}}$};
\node[rectangle,draw,xshift=10cm,yshift=-0.2cm] (24) {$L_{0},L_{1},L_{2}{\color{red}{D}}$};
\node[rectangle,draw,xshift=10cm,yshift=-1cm] (25) {$L_{0},L_{1}{\color{red}{D}},L_{2}{\color{red}{D}}$};
\node[rectangle,draw,xshift=10cm,yshift=-2cm] (26) {$L_{0},L_{1}{\color{red}{D}},L_{2}{\color{red}{D}}$};
		
\draw[dashed,->] (00) -- (11);
\draw[dashed,->] (00) -- (12);
\draw[dashed,->] (00) -- (13);
		
\draw[dashed,->] (11) -- (21);
\draw[dashed,->] (11) -- (22);
\draw[dashed,->] (11) -- (23);
		
\draw[dashed,->] (12) -- (24);
\draw[dashed,->] (12) -- (25);
\draw[dashed,->] (13) -- (26);
\end{tikzpicture}
\end{center}
\caption{Schematic presentation of the possible scenarios}
\label{Fig:Three_One}
\end{figure}
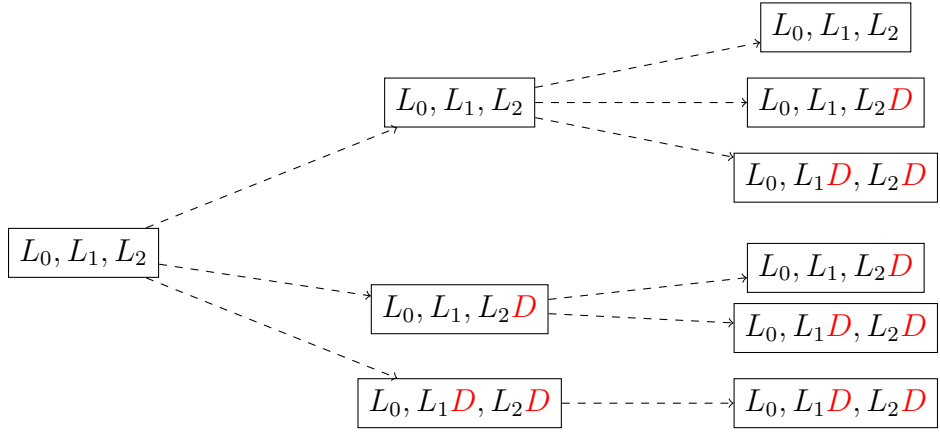

\begin{table}[h!]
\centering
\begin{tabular}{cc}
\hline
Variable & Description \\
\hline\hline
$X$ & Realization of the loan portfolio \\
$x$ & Portfolio of three assets \\
$\delta$ & Cost of equity \\
$\rho(x)$ & Risk of the loan portfolio $x$  \\
$\theta_{1}$ & Upper bound on risk at $t=0$ \\
$\theta_{2}$ & Upper bound on risk $t=1$ \\
$k_{lev}$ & Leverage Ratio restriction \\
$K(x)$ & Internal Ratings Based (IRB) capital requirement for portfolio $x$ \\
$e$ & Equity component of the bank's portfolio. \\
\hline
\end{tabular}
\caption{Description of the model variables}
\label{Tab:2.1}
\end{table}

We have constructed Model \ref{Model3.1} to determine the initial investment strategy, by considering the fact that the bank is subjected to limited liability \cite{Acosta20}, while at the same time meeting the regulatory capital requirements. The objective function maximizes the expected profit, after the payout to the debt-holders. It also incorporates the cost of equity. Here, the variable $\displaystyle{X=\sum\limits_{i=0}^{2}x_{i}L_{i}^{(1)}}$ is the random variable, where $L_{i}$ is the value of the $i$-th loan asset at time $t=1$. In case of default (the probability of which is $p_{i}$), $L_{i}$ takes the value $1-\lambda_{i}$. In the event of no default (the probability of which is $(1-p_{i})$), $L_{i}$ takes the value $(1+r_{i})$. It may be noted that we will denote $r_{0}=r$. Choosing a capital structure is also important, in the context of maximizing the payoff. Therefore, we have to maximize the objective function, with respect to $x$ and $e$.

\begin{model}[At $t=0$]
\label{Model3.1}	
\[\max_{x,e}\left[\mathbb{E}\left[\max\left(X-\beta_{ST}(1-e)-(1+r_{d}) \beta_{LT}(1-e),0\right)\right]-\delta e\right],\]
subject to the constraints,
\begin{enumerate}
\item $0 \leq x_{i} \leq 1~\forall~i=0,1,2$ (short selling is not permissible),
\item $\displaystyle{\sum\limits_{i=0}^{2} x_{i}=1}$,
\item $e \geq \max \left(k_{lev}, K(x)\right)$,
\item $\rho(x) \leq \theta_{1}$ (upper bound on risk).
\end{enumerate}
\end{model}
Here, $r_{d}$ is the interest rate being paid to the long term debt holders. This interest rate is lesser than the interest rate applicable for $L_{0}$, because banks make profit only when the interest rate that they charge on the loans is more than the interest rate that they pay the depositors. On the other hand, except safety bank also provides liquidity and transaction convenience to the depositors. Further, $\beta_{ST}$ denotes the fraction of depositors who are short term debt holders, and $\beta_{LT}(=1-\beta_{ST})$ is the fraction of long term debts issued by the bank. The modeling for banks decision problem at $t=1$ is a complex task, because it involves multiple possibilities. Accordingly, we enumerate the dynamics (at $t=1$) which needs to be incorporated into the model: 
\begin{enumerate}[(A)]
\item Bank can either raise money through debt and equity, or it can rebalance its assets to repay the short term debt and to meet regulatory requirements \cite{Behn2019}.
\item Bank will continue its business as long as the business is profitable.
\end{enumerate}

We have constructed a novel model to address this issue at the intermediate time step. This model has the advantage of liquidating and rebalancing assets, as well as raising money through debt and equity. In the model, if the equity holders expect to get more return for continuing their holdings into the next time point, then they  will hold on to their position until the next time point. Otherwise, they will liquidate their assets and collect the proceeds. We call thus as the ``equity holder's criteria''. Let us define $\tau_{i}$ as the default time for the $i$-th loan. Our model maximizes the expected return at time $t=1$ by fulfilling the capital requirement as well as the equity holder's criteria. The main contribution and novelty of our model is the inclusion of all the four strategies, suggested in \cite{Behn2019} (to maintain the Leverage Ratio) and the construction of which are enumerated below:  
\begin{enumerate}[(A)]
\item By issuing new equity (raised in the market or from retained earnings) to buy back debt, while keeping the assets constant.
\item By using new equity to fund asset growth, while keeping the debt constant.
\item By selling assets and using the proceeds to buy back debt, while keeping the equity constant.
\item By rebalancing assets towards less risky positions (thus decreasing average risk weights in the portfolios), while keeping assets and equity constant.	
\end{enumerate}
The last two strategies are included in the model via a term given by,
\[\mathbb{E}\left[\sum\limits_{i=0}^{2}x_{i}^{(1)}L_{i}^{(2)} \mathbbm{1}_{\tau_{i}>1}\right],\]
in the objective function, because it allows for rebalancing of assets in order to pay the debt holders. The variable $x^{(1)}=(x_{0}^{(1)},x_{1}^{(1)},x_{2}^{(1)})$, is the portfolio that is acquired at time $t=1$ and held till the final time point $t=2$. Let $v_{e}$ and $v_{d}$ the amount of new equity and debt (at $t=1$), respectively. It includes selling the assets without raising new equity and debt ($v_{e}=0=v_{d}$), which is incorporated in the third strategy. Further, the bank can also rebalance its assets to reduce the risk both with and without the condition $v_{e}=0=v_{d}$, after paying the short term debt also, thereby incorporating the fourth strategy. Our model allows to raise new debt and equity in the intermediate time step and also allow to restructure the loan portfolio. So the bank can issue new debt, keeping the asset constant and pay the short term debt, which is represented by the first strategy. Additionally, it can also simultaneously fund asset growth, which is captured in the second strategy.
\begin{model}[At $t=1$]
\label{Model3.2}
\[\max_{x^{(1)},v_{e},v_{d}} \left[\mathbb{E}\left[\sum\limits_{i=0}^{2} x_{i}^{(1)}L_{i}^{(2)} \mathbbm{1}_{\tau_{i}>1}\right]-\phi_{d}v_{d}-\phi_{e}v_{e}\right],\]
subject to the constraints of:
\begin{enumerate}
\item $\sum\limits_{i=0}^{2} \max\{x_{i}^{(1)}-x_{i},0\}L_{i}^{(1)} \mathbbm{1}_{\tau_{i}>1}+\text{ST}=\sum\limits_{i=0}^{2}\max\{x_{i}-x_{i}^{(1)},0\}L_{i}^{(1)}+(1-\phi_{e})v_{e}+(1-\phi_{d})v_{d}$,
\item Leverage Ratio $>$ Restriction,
\item $\displaystyle{\frac{{\frac{1}{1+r}} \mathbb{E}_{\mathbb{Q}}\left[ \left(\sum\limits_{i=0}^{2}x_{i}^{(1)}L_{i}^{(2)} \mathbbm{1}_{\tau_{i}>1}-v_{d}-\text{LT}\right)^{+}\right]}{e+v_{e}}\geq \min\left\{1+r_{e},\frac{\left(\sum\limits_{i=0}^{2} x_{i}L_{i}^{(1)}-\text{Debt}\right)^{+}}{e}\right\} \geq 0}$,
\item $\rho(x^{(1)}) \leq \theta_{2}$,
\item $ 0 \leq v_{e} \leq E$ and $ 0 \leq v_{d} \leq D$.
\end{enumerate}
\end{model}
In the model formulation, ST and LT, denote the short-term and long-term debt, respectively. Also, $r$ and $r_{e}$ is used to denote the risk-free rate and the target return on equity (for the equity holders to continue with their holdings), respectively. Further, $E$ and $D$ represent the upper bounds (or cap) on the amount of equity and debt (issued by the bank), respectively. Finally, $\phi_{e}$ and $\phi_{d}$, are the cost of issuing equity and debt, respectively. We now make the following observations with regards to the constraints, incorporated in the model.
\begin{enumerate}[(A)]
\item The objective function maximizes the return of expected payoff at final time step and minimizes the cost of raising capital.
\item The first constraint represents the budget restriction at the intermediate time step. It represents that the money used to repay the short-term debt and to invest more in loans is equal to the funds raised through debt, equity and selling some of its assets. The bank is able to make an investment if the loan has not defaulted yet. Therefore we have added the term $\mathbbm{1}_{\tau_{i}>1}$, where $\tau_{i}$ is the default time for the $i$-th asset.
\item The second constraint shows that at time $t=1$, the bank must satisfy the Leverage Ratio restriction.
\item The third constraint implies that the expected (with respect to risk-neutral probability measure $Q$) discounted return on equity at time $t=1$ should be greater than the minimum value between the target growth factor and the return on equity at time $t=1$. Shareholders allow for continuation of their current position, on to the next time step, provided that the constraint is satisfied. Else, they will liquidate their position. Therefore (as already noted) we call it as ``equity holder's constraint''. Here, we have compared the present return on equity with the expected discounted return on equity. For this valuation, we have considered risk-neutral probability measures motivated by the American call option pricing technique \cite{Cutland2012}.
\item The fourth constraint creates an upper bound on the risk for the loan portfolio at time $t=1$.
\item Finally, the fifth constraint represents that in any situation, there is an cap on the amount of equity and debt raised, respectively.
\end{enumerate}

\section{Theoretical Results}
\label{Sec_Three_Theoretical Results}

In this Section, we discuss some theoretical results which follows from the models presented in the preceding Section. In \cite{Barik2023}, the authors have studied the scenario where a bank cannot raise capital through debt and equity, in the intermediate time step ($t=1$). In this article, we study a general setup, where banks can issue debt and equity, as well as liquidate their assets. Behn et al. \cite{Behn2019} have presented a dynamic bank model to study the bank's behavior in order to meet the regulatory constraints. One of the most important regulatory aspects is the capital structure'. The regulatory framework, formulated by the Basel Committee on Banking Supervision (BCBS) had introduced the ``Leverage Ratio'', to facilitate the resilience of the banking system. Accordingly, the authors in \cite{Behn2019} have stated four strategies to fulfill the capital requirement. In this study, we have constructed a novel model by incorporating these four strategies via simple Model \ref{Model3.2}. It maximizes the expected return and minimizes the cost of issuing new debt or equity. It also considers the capital ratio requirement and equity holders' constraints. In the following description, we describe how Model \ref{Model3.2} incorporates all these four strategies.

According to the modeling setup, the bank can raise capital through debt, equity and liquidating assets, but the question is to ascertain as to which of these is the most preferable for a bank? Are these decisions scenario-dependent (the nodes in Figure \ref{Fig:Three_One} represent the scenario)? We have addressed these cases in the following results. Before going to the discussion of the statements of the theorems and their proofs, we have made an assumption for the brevity of the proof and to obtain a relation between the issuance of debt and equity simultaneously (the assumption is that all the realizations are non-negative at the final time step). 
\begin{assumption}
\label{Assumption 1}
\begin{eqnarray*}
&&\frac{{\frac{1}{1+r}} \mathbb{E}_{\mathbb{Q}} \left[\left(\sum\limits_{i=0}^{2}x_{i}^{(1)}L_{i}^{(2)} \mathbbm{1}_{\tau_{i}>1}-v_{d}-LT\right)^{+}\right]}{e+v_{e}}\\
&\geq&\frac{{\frac{1}{1+r}} \mathbb{E}_{\mathbb{Q}} \left[\left(\sum\limits_{i=0}^{2} x_{i}^{(1)}L_{i}^{(2)} \mathbbm{1}_{\tau_{i}>1}-v_{d}-LT\right)\right]}{e+v_{e}}\\ 
&\geq& \min\left\{1+r_{e},\frac{ \left(\sum_{i=0}^{2} x_{i}L_{i}^{(1)}-\text{Debt}\right)^{+}}{e}\right\}.
\end{eqnarray*}
\end{assumption}
When all the realization are positive then, the first term in the above relation becomes positive.
\begin{theorem}
\label{Theorem:3.1}
Bank prefers to raise capital from that source (debt and equity) which has the lower raising cost.
\end{theorem}
In Model \ref{Model3.1}, we have considered that there are two types of debt, namely long term debt and short term debt. Short term debt affects the solvency of the bank, in the intermediate time step.  
\begin{theorem}
Larger amount of short term debt at initial time point increases the chance of default at the intermediate time step.
\end{theorem}
\begin{proof}
The bank faces default in the intermediate time step, if it is unable to meet the depositors' claim, that is, the bank is unable to return the money back to the short term debt holders. This happens if the realization from portfolio is less than the value of short term debt. So, if $R_{1}$ is the realization of the portfolio at time $t=1$, then the bank survives the worst case (Node-3) provided the following condition holds:
\[R_{1}\left[\text{For the worst case}\right] \geq \beta_{ST}(1-e)\Rightarrow 
x_{0}(1+r)+x_{1}(1-\lambda_{1})+x_{2}(1-\lambda_{2}) \geq \beta_{ST}(1-e).\]
This implies that,
\[\beta_{ST} \le \frac{x_{0}(1+r)+x_{1}(1-\lambda_{1})+x_{2}(1-\lambda_{2}) }{(1-e)}.\]
From the above inequalities, we can clearly see that for a large value of $\beta_{ST}$, the bank has lesser chance of surviving at $t=1$. Therefore, it increases the chance of default. The other nodes also behave in a similar manner. 
\end{proof}
Having fewer short-term debts has some advantages for the other nodes. In Node-2, if $x_{2}(1-\lambda_{2}) \geq \beta_{ST}(1-e)$, then the bank can meet depositors' claims without liquidating assets or issuing new debt or equity. Also, at Node-1, a lesser amount of short-term obligation reduces the amount of liquidation or issuing new debt and equity, which, in turn reduces the cost of raising money in the intermediate time step. Before proceeding further, we introduce the variable $Z_{1}$ as the realization, after paying the short-term debt, that is, $Z_{1}=R_{1}-\text{Short term debt}=R_{1}-\beta_{ST}(1-e)$.
\begin{theorem}
\label{Theorem:3.3}
Bank can survive the worst scenario at $t=1$, provided, $\displaystyle{x_{0}=\max\left\{0,\frac{x_{1} \lambda_{1}+x_{2}\lambda_{2}+\beta_{ST}(1-e)-1}{r}\right\}}$.  
\end{theorem}
\begin{proof}
Let us assume that the bank survives Node-3, at $t=1$. Then the net value of the assets is more than the debt, that is,
\begin{eqnarray*}
&&x_{0}(1+r)+x_{1}(1-\lambda_{1})+x_{2}(1-\lambda_{2}) \geq \beta_{ST}(1-e)\\
&\Rightarrow& x_{0}r \geq x_{1} \lambda_{1}+x_{2} \lambda_{2}+\beta_{ST}(1-e)-1\\
&\Rightarrow& x_{0} \geq \frac{x_{1} \lambda_{1}+x_{2} \lambda_{2}+\beta_{ST}(1-e)-1}{r}.
\end{eqnarray*}
This relation says that the investment, $x_{0}$ satisfying the above inequality can help the bank to survive the worst scenario, if all the other conditions are satisfied.
\end{proof}
The authors in \cite{Acosta20} have discussed the effect of the Leverage Ratio on bank stability, by considering different states ($s_{1}$ and $s_{2}$) to study the relation. In our study, the scenarios are presented in Figure \ref{Fig:Three_One}, and each scenario presents a state. The study in \cite{Acosta20} shows that more equity leads to stable banking. However, having more equity reduces the return on equity, and we prove this argument in the following theorem.
\begin{theorem}
In the intermediate time step, the return on equity reduces, as the amount of equity holders increases.
\end{theorem}
\begin{proof}
Let $\mathbb{E}[R_{1}]$ denote the expectation of the random variable $R_{1}$, presenting the realization of the loan portfolio at time $t=1$. Then the expected return on equity is given by,
\begin{eqnarray*}
r_{\text{equit}y}&:=&\frac{\mathbb{E}[R_{1}]-\beta_{ST}(1-e)-\beta_{LT}(1+r_{d})(1-e)}{e}\\
&=& \frac{\mathbb{E}[R_{1}]-(1-e)(\beta_{ST}+\beta_{LT}(1+r_{d}))}{e} \\
&=& \frac{\mathbb{E}[R_{1}]-(1-e)(1+\beta_{LT} r_{d})}{e} \\
&=& \frac{\mathbb{E}[R_{1}]-(1+\beta_{LT}r_{d})+e(1+\beta_{LT}r_{d})}{e}.
\end{eqnarray*}
Now differentiating with respect to $e$, we get,
\begin{eqnarray*}
\frac{\partial r_{\text{equity}}}{\partial e}&=& \frac{e(1+\beta_{LT}r_{d})-(\mathbb{E}[R_{1}]-(1+\beta_{LT}r_{d})+e(1+\beta_{LT}r_{d}))}{e^{2}} \\
&=&\frac{-(\mathbb{E}[R_{1}]-(1+\beta_{LT}r_{d}))}{e^{2}}.
\end{eqnarray*}
Now in order to show that more amount of equity reduces the expected return on equity, we need to show that $\displaystyle{\frac{\partial r_{\text{equity}}}{\partial e} \leq 0}$, and to prove this, it only remains to show that $\mathbb{E}[R_{1}] \geq (1+\beta_{LT}r_{d})$. From the definition we have,
\begin{eqnarray*}
\mathbb{E}[R_{1}]&=&x_{0}\mathbb{E}[L_{0}^{(1)}]+x_{1}\mathbb{E}[L_{1}^{(1)}]+x_{2}\mathbb{E}[L_{2}^{(1)}] \\
&\geq& \mathbb{E}[L_{0}^{(1)}] 
~~~~~~(\text{because}~{\mathbb{E}[L_{2}^{(1)}] \geq \mathbb{E}[L_{1}^{(1)}] \geq \mathbb{E}[L_{0}^{(1)}]}) \\
&=& (1+r) \\
&=& \beta_{ST}(1+r)+\beta_{LT}(1+r) \\
&\geq& \beta_{ST}+\beta_{LT}(1+r) \\
&\geq& \beta_{ST}+\beta_{LT}(1+r_{d})~~~~~~(\text{as}~r_{d} \leq r)\\ 
&=&1+\beta_{LT}r_{d}
\end{eqnarray*}
This completes the proof.
\end{proof}
The preceding theorem discusses about the expected return on equity at the intermediate time step. Next, we study the change of return on equity with the amount raised at the intermediate time step, through equity $v_{e}$ and debt $v_{d}$. Let $V_{\text{equity}}(v_{e},v_{d})$ denotes the expected discounted return on equity, while raising $v_{e}$ amount of equity and $v_{d}$ amount of debt. Using Assumption \ref{Assumption 1}, we obtain,
\[V_{\text{equity}}(v_{e},v_{d})=\frac{Z_{1}+(1-\phi_{e})v_{e}+(1-\phi_{d})v_{d}-\frac{v_{d}+\beta_{LT}(1-e)(1+r_{d})^{2}}{1+r}}{e+v_{e}}.\] 
In Figure \ref{Fig:Three_ROE} we demonstrate the change of $V_{\text{equity}}(v_{e},v_{d})$ versus $v_{d}$($v_{e}=0$) and $v_{e}$($v_{d}=0$).
\begin{figure}[h!]
\begin{center}   
\begin{tikzpicture}
\begin{axis}[
width=7cm, height=7cm,
at={(-500,0)}, 
xlabel={$v_{d}$},
ylabel={$V_{\text{equity}}(0,v_{d})$},
domain=0:0.5, 
samples=100,
axis lines=middle, 
smooth
]
\addplot[
thick,
blue,
]
{(0.055 + (1-0.05-1/1.09)*x)/(0.06)};
\end{axis}

\begin{axis}[
width=7cm, height=7cm,
at={(500,0)}, 
xlabel={$v_{e}$},
ylabel={$V_{\text{equity}}(v_{e},0)$},
domain=0:0.5, 
samples=100,
axis lines=middle, 
smooth
]
\addplot[
thick,
blue,
]
{(0.055 + 0.09*x)/(0.06 + x)};
\end{axis}
\end{tikzpicture}
\end{center}
\caption[Return on equity with $v_{d}$ and $v_{e}$]{This figure shows the change of return on equity with issuing new debt and equity in the intermediate time step $\left(\displaystyle{\phi_{d} \leq \frac{r}{1+r}}\right)$.}
\label{Fig:Three_ROE}
\end{figure}
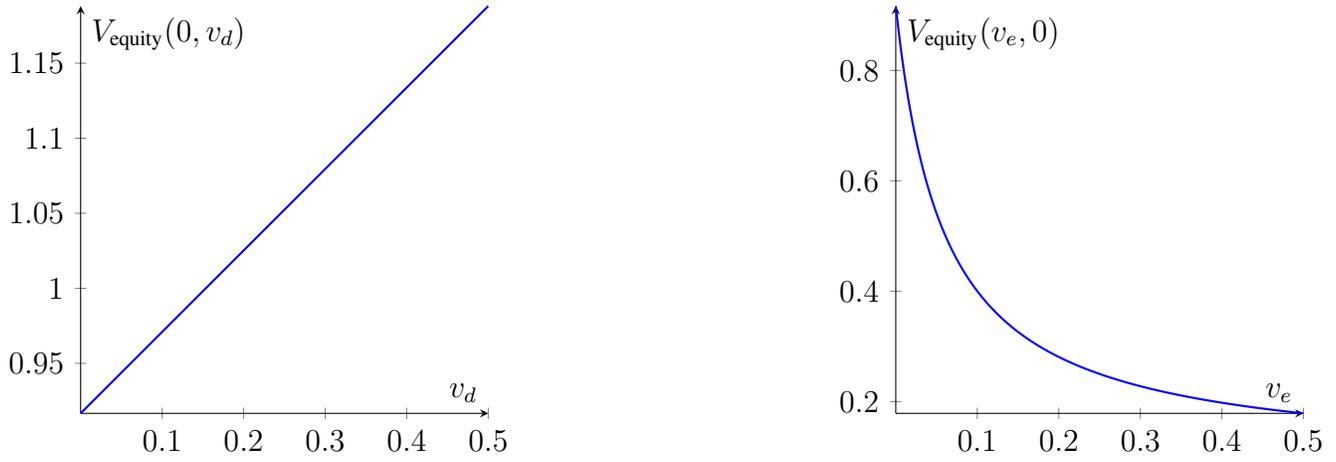
In the next theorem we study the change of $V_{\text{equity}}$ with $v_{d}$, without raising any equity ($v_{e}=0$).
\begin{theorem}
Issuing fund only through debt (without equity) increase the discounted expected payoff of the equity, provided $\displaystyle{\phi_{d} \leq \frac{r}{1+r}}$.
\end{theorem}
\begin{proof}
Let $R_{1}$ denote the realization of the portfolio at time $t=1$. Let us consider that the bank has issued $v_{d}$ amount of debt. Then it has paid $\phi_{d}v_{d}$ to raise this amount, and it has the same amount of equity $e$. Then,
\[V_{\text{equity}}(0,v_{d})=\frac{Z_{1}+(1-\phi_{d})v_{d}-\frac{v_{d}}{1+r}+\frac{\beta_{LT} (1-e)(1+r_{d})^{2}}{1+r}}{e}.\]
Here $Z_{1}=R_{1}-\beta_{ST}(1-e)$ is the realization after paying the short term debt. We get the value of the equity by discounted expected realization of equity at the final time point, with the expectation being taken with respect to risk neutral probability measure. To get the value of equity per unit volume, we divide it by the amount of equity. Now differentiating w.r.t. $v_{d}$, we get,
\[\frac{\partial}{\partial v_{d}}V_{\text{equity}}(0,v_{d})=\frac{1-\phi_{d}-\frac{1}{1+r}}{e}.\]
Hence $\displaystyle{\frac{\partial V_{equity}}{\partial v_{d}}}$ will be an increasing function, provided,
\[1-\phi_{d}-\frac{1}{1+r} \geq 0 \Rightarrow \frac{r}{1+r} \geq \phi_{d}.\] 
This completes the proof.
\end{proof}
At the time $t=1$ and Node-3, for Model \ref{Fig:Three_One}, the two loans have already defaulted. If the bank is able to survive this scenario, it can issue new debt to increase profit for the equity-holders, enabling them to continue their holding to the next time step.
\begin{corollary}
For $\displaystyle{\phi_{d} \geq \frac{r}{1+r}}$, $V_{\text{equity}}(0,v_{d})$ is decreasing function. Then the bank has an upper bound on raising new debt from equity holders constraint. 
\end{corollary}
\begin{proof}
The proof of the first part is quite obvious from the preceding theorem. Now to get the cap, we use the equity holder's constraint in the Model \ref{Model3.2},
\begin{enumerate}[(A)]
\item When the expected return is less than the target, then,
\begin{eqnarray*}
&&Z_{1}+(1-\phi_{d})v_{d}-\frac{v_{d}}{1+r}-\beta_{LT}(1-e)\frac{(1+r_{d})^{2}}{1+r}\geq Z_{1}-\beta_{LT}(1-e)(1+r_{d}) \\
&\Rightarrow& \left(1-\phi_{d}-\frac{1}{1+r}\right)v_{d} \geq \beta_{LT}(1-e)(1+r_{d})\left(\frac{1+r_{d}}{1+r}-1\right) \\
&\Rightarrow& \left(\frac{r}{1+r}-\phi_{d}\right)v_{d} \geq \beta_{LT}(1-e)(1+r_{d})\left(\frac{r_{d}-r}{1+r}\right).
\end{eqnarray*}
From our construction, we have, $r_{d}<r$, and therefore multiplying by $-1$, we get,
\begin{eqnarray*}
&&\left(\phi_{d} - \frac{r}{1+r}\right)v_{d} \leq \beta_{LT}(1-e)(1+r_{d})\left(\frac{r-r_{d}}{1+r}\right) \\
&\Rightarrow& v_{d} \leq \frac{\beta_{LT}(1-e)(1+r_{d})\left(\frac{r-r_{d}}{1+r}\right)}{\left(\phi_{d}-\frac{r}{1+r}\right)}.
\end{eqnarray*}
\item For the case when the return exceeds the target return, we have,
\begin{eqnarray*}
&&\frac{Z_{1}+(1-\phi_{d})v_{d}-\frac{v_{d}}{1+r}-\beta_{LT}(1-e)\frac{(1+r_{d})^{2}}{1+r}}{e} \geq (1+r_{e})\\
&\Rightarrow& Z_{1}-\beta_{LT}(1-e)\frac{(1+r_{d})^{2}}{1+r}-(1+r_{e})e \geq \left(\phi_{d}-\frac{r}{1+r}\right)v_{d} \\
&\Rightarrow& v_{d} \leq \frac{Z_{1}-\beta_{LT}(1-e)\frac{(1+r_{d})^{2}}{1+r}- (1+r_{e})e}{\left(\phi_{d}-\frac{r}{1+r}\right)v_{d}}.
\end{eqnarray*}
\end{enumerate}
The above inequalities gives an upper bound on how much debt the bank can issue when the raising cost is high.
\end{proof}
Raising money through debt decreases the net debt of the bank, and hence it reduces the Leverage Ratio.
\begin{theorem}
Issuing money through debt decreases the Leverage Ratio.
\end{theorem}
\begin{proof}
For raising $v_{d}$ amount of debt, the Leverage Ratio is given by,
\[\text{Leverage Ratio}:=\frac{Z_1+(1-\phi_{d})v_{d}-v_{d}-\beta_{LT}(1-e)(1+r_{d})}{Z_{1}+v_{d}}.\]
Partially differentiating with respect to $v_{d}$, we get,
\begin{eqnarray*}
\frac{\partial}{\partial v_{d}}\left(\text{Leverage Ratio}\right)&=& \frac{-\phi_{d}(Z_{1}+v_{d})-(Z_1-\phi_{d}v_{d}-\beta_{LT}(1-e)(1+r_{d}))}{(Z_{1}+v_{d})^{2}}\\
&=& \frac{-Z_{1}(1+\phi_{d})+\beta_{LT}(1-e)(1+r_{d})}{(Z_{1}+v_{d})^{2}}
\end{eqnarray*}
If the bank survives, then its realization must be more than the debt, that is,
\[Z_{1} \geq \beta_{LT}(1-e)(1+r_{d}) \Rightarrow Z_{1}(1+\phi_{d}) \geq \beta_{LT}(1-e)(1+r_{d}).\]
The last inequality implies that $\displaystyle{\frac{\partial}{\partial v_{d}}\left(\text{Leverage Ratio}\right)}$ is negative. Hence the Leverage Ratio decreases with $v_{d}$.
\end{proof}
Issuing new debt increases leverage, and so the Leverage Ratio restriction creates a cap for the amount of debt issued.
\begin{theorem}
Leverage Ratio restriction implements a cap on the amount of debt can be issued.
\end{theorem}
\begin{proof}
The capital ratio constraint, given by second constraint in Model \ref{Model3.2}, is given by,
\begin{eqnarray*}
&&\frac{Z_{1}+(1-\phi_{d})v_{d}-v_{d}-\beta_{LT}(1-e)(1+r_{d})}{Z_{1}+v_{d}}\geq k_{lev} \\
&\Rightarrow& Z_{1}+(1-\phi_{d})v_{d}-v_{d}-\beta_{LT}(1-e)(1+r_{d}) \geq k_{lev} Z_{1}+k_{lev} v_{d} \\
&\Rightarrow& Z_{1}-k_{lev} Z_{1}-\beta_{LT}(1-e)(1+r_{d}) \geq v_{d}(k_{lev}+\phi_{d}) \\
&\Rightarrow& v_{d} \leq \frac{Z_1-k_{lev} Z_{1}-\beta_{LT}(1-e)(1+r_{d})}{k_{lev}+\phi_{d}}.
\end{eqnarray*}
The above inequality defines a cap through capital requirement condition. It can be easily seen that this upper bound increases as $k_{lev}$ and $\phi_{d}$ decrease, that is, for higher value of Leverage Ratio, the restriction decreases the cap on the debt that can be issued. For higher values of $Z_{1}$, that is, the realization of the portfolio acquired at time $t=0$ allows more debt to be issued.
\end{proof}
Therefore, from the above two results, we conclude that the model allows the bank to issue debt that satisfy both the equity holders constraint and capital requirement condition. In the next case we show the effect when the firm issue new equity only. Now, if the bank issues new equity, then the expected discounted return per unit share is,
\begin{eqnarray*}
V_{\text{equity}}(v_{e},0)&=&\frac{\frac{1}{1+r}\left[\mathbb{E}_{Q} \left[(Z_{1}+(1-\phi_{e})v_{e})(1+r)-(1+r_{d})^{2}\beta_{LT}(1-e)\right] \right]}{e+v_{e}} \\
&=&\frac{Z_{1}-\frac{(1+r_{d})^{2}}{1+r}\beta_{LT}(1-e)+(1-\phi_{e})v_{e}}{e+v_{e}}.
\end{eqnarray*}
This definition leads us to study the change of the function, $f(v_{e})$. The following theorem discusses this relation mathematically. Let $\displaystyle{Z_{1LT}:=Z_{1}-\frac{(1+r_{d})^{2}}{1+r}\beta_{LT}(1-e)}$, that is, the discounted return on equity at time $t=1$.
\begin{theorem}
If $Z_{1LT} \geq (1-\phi_{e})e$, then $V_{\text{equity}}(v_{e},0)$ decreases, as $v_{e}$ increases.                                                                                                                                                            
\end{theorem}
\begin{proof}	
We can define the following function, in terms of $Z_{1LT}$.
\[V_{\text{equity}}(v_{e},0)=\frac{Z_{1LT}+(1-\phi_{e})v_{e}}{e+v_{e}}.\]
Differentiating w.r.t., $v_{e}$, we have, 
\[\frac{\partial}{\partial v_{e}}V_{\text{equity}}(v_{e},0)=\frac{(1-\phi_{e})e-Z_{1LT}}{(e+v_{e})^{2}} \leq 0.\]
Since, $Z_{1LT} \geq e$, therefore, $Z_{1LT} \geq (1-\phi_{e})e$.	
\end{proof}
This theorem implies that in favorable scenarios, raising new equity decreases the expected discounted return on equity. 
\begin{corollary}
\label{Corollary:2}
In some adverse scenarios, the reverse can happen, that is, if $Z_{1LT} \leq (1-\phi_{e})e$ then $V_{\text{equity}}(v_{e},0)$ increases with an increase in $v_{e}$.
\end{corollary}
Issuing new equity reduces the return on equity. Hence from the equity holders constraint, we get an upper bound for the amount of equity we can issue.
\begin{theorem}
From equity holder's constraint we get an upper bound on the amount of equity issued. Equity holder’s constraint is more than the minimum value of the current return or the target return.
\end{theorem}
\begin{proof}
From the equity holder's constraint, we can say that the return on equity should be more than the present value. 
\begin{enumerate}[(A)]
\item When the current return on equity is less than the target return, then,
\begin{eqnarray*}
&&\frac{Z_{1}+(1-\phi_{e})v_{e}- \beta_{LT}(1-e)\frac{(1+r_{d})^{2}}{1+r}}{e+v_{e}} \geq \frac{ Z_{1}-\beta_{LT}(1-e)(1+r_{d})}{e} \\
&\Rightarrow&eZ_{1}+e(1-\phi_{e})v_{e}-e\beta_{LT}(1-e)\frac{(1+r_{d})^{2}}{1+r} \\
&\geq& eZ_{1}-e\beta_{LT}(1-e)(1+r_{d})+v_{e}Z_{1}-v_{e}\beta_{LT}(1-e)(1+r_{d})\\
&\Rightarrow& e\beta_{LT}(1-e)(1+r_{d})-e\beta_{LT}(1-e)\frac{(1+r_{d})^{2}}{1+r} \geq v_{e} \left(Z_{1}-\beta_{LT}(1-e)(1+r_{d})-e(1-\phi_{e})\right) \\
&\Rightarrow & v_{e} \leq \frac{\beta_{LT}e(1-e)(1+r_{d})\frac{r-r_{d}}{1+r}}{Z_{1}-\beta_{LT}(1-e)(1+r_{d})-(1-\phi_{e})e}.
\end{eqnarray*}
\item When the current return crosses the target return, then,
\begin{eqnarray*}
&&\frac{Z_{1}+(1-\phi_{e})v_{e}-\beta_{LT}(1-e)\frac{(1+r_{d})^{2}}{1+r}}{e+v_{e}} \geq (1+r_{e}) \\
&\Rightarrow& Z_{1}+(1-\phi_{e})v_{e}-\beta_{LT}(1-e)\frac{(1+r_{d})^{2}}{1+r}\geq (1+r_{e})e+(1+r_{e})v_{e} \\
&\Rightarrow& Z_{1}-\beta_{LT}(1-e)\frac{(1+r_{d})^{2}}{1+r}-(1+r_{e})e \geq v_{e}(r_{e}+\phi_{e}) \\
&\Rightarrow& v_{e} \leq \frac{Z_{1}-\beta_{LT}(1-e)\frac{(1+r_{d})^{2}}{1+r}-(1+r_{e})e}{(r_{e}+\phi_{e})}.
\end{eqnarray*}
\end{enumerate}
The above inequalities give the upper bound on the amount of capital that can be raised via equity.
\end{proof}
For lower issuing cost ($\phi_{e}$), issuing new equity raises the capital ratio.
\begin{theorem}
\label{Theorem:3.10}
If $\displaystyle{\phi_{e} \leq \frac{\beta_{LT}(1-e)(1+r_{d})}{Z_{1}}}$, then the Leverage Ratio increases, with issuing new equity.
\end{theorem}
\begin{proof}
After issuing $v_{e}$ amount of capital, the Leverage Ratio is given by,
\[\text{Leverage Ratio}:=\frac{Z_{1}+(1-\phi_{e})v_{e}-\beta_{LT}(1-e)(1+r_{d})}{Z_{1}+v_{e}}.\]
Now differentiating it w.r.t. $v_{e}$, we get,
\[\frac{\partial}{\partial v_{e}} \left(\text{Capital Ratio}\right)=\frac{\beta_{LT}(1-e)(1+r_{d})- Z_{1} \phi_{e}}{(Z_{1}+v_{e})^{2}}.\]
Clearly if the given condition is satisfied then capital ratio is an increasing function of $v_{e}$.
\end{proof}
Clearly if bank has less debt and the issuing cost is high, then it not necessary to collect fund through equity. The payment to issue new equity then becomes a new burden. The changing pattern of the Leverage Ratio with $v_{e}$ and $v_{d}$ is presented in Figure \ref{Fig:Three_Caprat}.
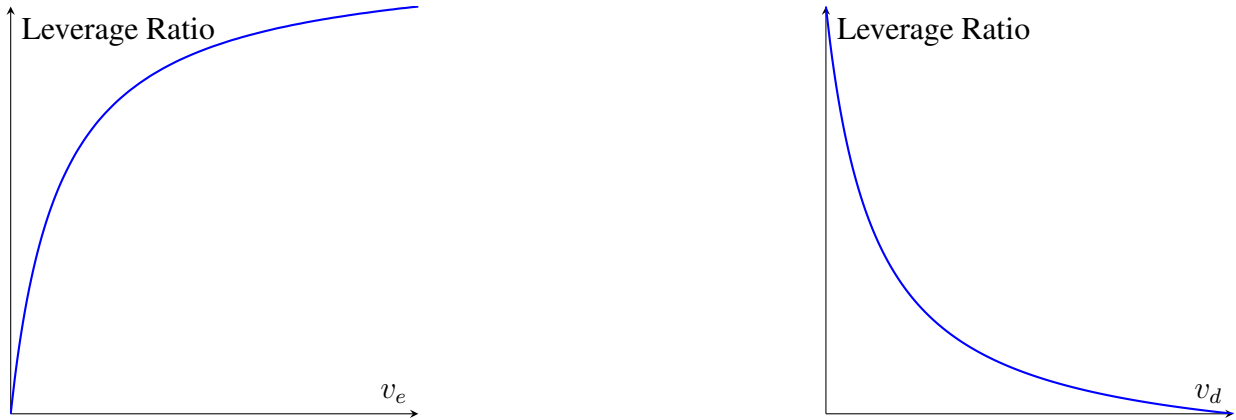
\begin{figure}[ht]
\vspace{0.5cm}
\begin{center}   
\begin{tikzpicture}
\begin{axis}[
width=7cm, height=7cm,
at={(-500,0)}, 
xlabel={$v_{e}$},
ylabel={Leverage Ratio},
domain=0:0.5, 
samples=100,
axis lines=middle, 
xtick=\empty, 
ytick=\empty,  
smooth
]
\addplot[
thick,
blue,
]
{(0.05 + 0.9*x)/(0.06 + x)};
\end{axis}
		
\begin{axis}[
width=7cm, height=7cm,
at={(500,0)}, 
xlabel={$v_{d}$},
ylabel={Leverage Ratio},
domain=0:0.5, 
samples=100,
axis lines=middle, 
xtick=\empty, 
ytick=\empty,  
smooth
]
\addplot[
thick,
blue,
]
{(0.05 - 0.05*x)/(0.06 + x)};
\end{axis}
\end{tikzpicture}
\end{center}
\caption[Capital Ratio with $v_{e}$]{This figure shows the change of capital ratio with issuing new equity and debt in the intermediate time step.}
\label{Fig:Three_Caprat}
\end{figure}
The above results discuss the cap on the capital raised through a single source, namely, either debt or equity. All the cap are in terms of the variable introduced initially and the realization of the portfolio implemented at the initial stage. These caps do not contain any decision variable of the intermediate time step ($x^{(1)},v_{d},v_{e}$). Now, we proceed to the case where banks issue funds through both the sources, debt and equity. Here, we also find that the caps, for the variable $v_{e}$ and $v_{d}$, are dependent on the decision variables of the intermediate time step. Hence,considering all the advantages, the bank must satisfy all the constraints defined in the intermediate time step.
\begin{theorem}
\label{Theorem:3.11}
Leverage Ratio restriction gives an upper bound on the amount of debt that can be issued, which is dependent on the amount of equity which has been raised.
\end{theorem}
\begin{proof}
From the Leverage Ratio requirement, we get, 
\[\frac{Z_{1}+(1-\phi_{d})v_{d}+(1-\phi_{e})v_{e}-v_{d}-\beta_{LT}(1-e)(1+r_{d})}{Z_{1}+v_{d}+v{e}} \geq k_{lev},\]
which implies that,
\begin{eqnarray*}
&&Z_{1}+(1-\phi_{d})v_{d}+(1-\phi_{e})v_{e}-v_{d}-\beta_{LT}(1-e)(1+r_{d}) \geq k_{lev} Z_{1}+k_{lev} v_{d} + k_{lev}v_{e} \\
&\Rightarrow& Z_{1}-k_{lev} Z_{1}-\beta_{LT}(1-e)(1+r_{d})+(1-\phi_{e})v_{e}-k_{lev}v_{e} \geq v_{d}(k_{lev}+\phi_{d}) \\ 
&\Rightarrow& v_{d} \leq \frac{Z_1-k_{lev} Z_{1}-\beta_{LT}(1-e)(1+r_{d})+(1-\phi_{e}-k_{lev})v_{e} } {k_{lev}+\phi_{d}} \\
\end{eqnarray*}
The above inequality imposes an upper bound  on $v_{d}$ that can be issued when $v_{e}$ amount of equity has been raised.
\end{proof}
\begin{theorem}
\label{Theorem:3.12}
Equity holder's constraint gives an upper bound on the amount of equity that can be raised. This upper bound is dependent of the amount of debt being issued.
\end{theorem}
\begin{proof} The facility of raising capital through debt and equity simultaneously, affect each other. The upper bound on the amount of equity that can be raised is dependent on the amount of debt issued.
\begin{enumerate}[(A)]
\item When the current return on equity is less than the target return, then we have,
\begin{eqnarray*}
&&\frac{Z_{1}+(1-\phi_{e})v_{e}+(1-\phi_{d})v_{d}-\frac{v_{d}+\beta_{LT}(1-e)(1+r_{d})^{2}}{1+r}}{e+v_{e}} \geq \frac{ Z_{1}-\beta_{LT}(1-e)(1+r_{d})}{e} \\
&\Rightarrow& eZ_{1}+e(1-\phi_{e})v_{e}+e(1-\phi_{d})v_{d}-\frac{ev_{d}}{1+r}-e\beta_{LT}(1-e) \frac{(1+r_{d})^{2}}{1+r} \\
&\geq& eZ_{1}-e\beta_{LT}(1-e)(1+r_{d})+v_{e}Z_{1}-v_{e}\beta_{LT}(1-e)(1+r_{d}) \\ 
&\Rightarrow& ev_{d}\left(1-\phi_{d}-\frac{1}{1+r}\right)+e\beta_{LT}(1-e)(1+r_{d})-e\beta_{LT}(1-e)\frac{(1+r_{d})^{2}}{1+r} \\
&\geq& v_{e}\left(Z_{1}-\beta_{LT}(1-e)(1+r_{d})-e(1-\phi_{e})\right) \\
&\Rightarrow& v_{e} \leq \frac{ev_{d}\left(1-\phi_{d}-\frac{1}{1+r}\right)+\beta_{LT}e(1-e)(1+r_{d})\frac{r-r_{d}}{1+r}}{Z_{1}-\beta_{LT}(1-e)(1+r_{d})-e(1-\phi_{e})}.
\end{eqnarray*}
\item When the current return crosses the target return, we obtain,
\begin{eqnarray*}
&&\frac{Z_{1}+(1-\phi_{e})v_{e}+(1-\phi_{d})v_{d}-\frac{v_{d}+\beta_{LT}(1-e)(1+r_{d})^{2}}{1+r}}{e+v_{e}} \geq \left(1+r_{e}\right) \\
&\Rightarrow & Z_{1}+\left(1-\phi_{e}\right)v_{e}+\left(1-\phi_{d}-\frac{1}{1+r}\right)v_{d}-\beta_{LT}(1-e) \frac{(1+r_{d})^{2}}{1+r} \geq (1+r_{e})(e+v_{e}) \\
&\Rightarrow & v_{e}(\phi_{e}+r_{e})\leq Z_{1}+\left(\frac{r}{1+r}-\phi_{d}\right)v_{d}-\beta_{LT}(1-e)\frac{(1+r_{d})^{2}}{1+r}-(1+r_{e})e \\
&\Rightarrow & v_{e} \leq \frac{Z_{1}+\left(\frac{r}{1+r}-\phi_{d}\right)v_{d}-\beta_{LT}(1-e)\frac{(1+r_{d})^{2}}{1+r}-(1+r_{e})e}{\phi_{e}+r_{e}}.
\end{eqnarray*}
\end{enumerate}
This inequalities create a cap on $v_{e}$ when $v_{d}$ is already issued, and the cap increases with $v_{d}$.
\end{proof}
From the above inequality, it is clear that in favorable condition, the value of the cap gets reduced, due to a higher value of $Z_{1}$.
Finally, the results regarding the issuance of debt and equity are summarized in Table \ref{Table:Three_One}.
\begin{table}[ht]
\caption{Table to summarize the results. $^{*}$ Decreases provided $Z_{1LT}\ge (1-\phi_{e})e$. $^{**}$ Increases provided $Z_{1LT} < (1-\phi_{e})e$. $^{\dag}$ No, provided $\displaystyle{\phi_{d}\le \frac{r}{1+r}}$, $^{\dag\dag}$ Yes, provided $\displaystyle{\phi_{d}> \frac{r}{1+r}}$.}
\label{Table:Three_One}
\begin{center}
\begin{tabular}{|c|c|c|c|c|}
\hline
& Change in &  Cap due to & Change in  & Cap due to `Equity\\
&Capital Ratio	& Capital Ratio & $V_{equity}$	& holder's constraint'. \\
\hline
Only, & Decreases &Yes & Increases & No $^{\dag}$  \\
$v_{d}$	& &&& Yes $^{\dag\dag}$ \\
\hline
Only, & Increases & No & Decreases $^{*}$ & Yes \\
$v_{e}$	& & & Increases $^{**}$ &\\
\hline
Both $v_{e}$, & Simultaneous  & Cap on $v_{d}$ & Simultaneous & Cap on $v_{e}$ \\
$v_{d}$	& effect & dependent on $v_{e}$ & effect & dependent on $v_{d}$\\
\hline
\end{tabular}
\end{center}
\end{table}

\section{An Example}
\label{Sec_Three_An_Example}

In this Section, we present an example to validate the preceding results. In particular, we consider two examples, one with parameter values of $\beta_{ST}=0.7,~D=0.01,~E=0.02$ and the other with the parameter values $\beta_{ST}=0.5,~D=0.20,~E= 0.30$. The remaining parameters which are identical for both the examples, are given in Tables \ref{Tab_Three_Loan_Parameters} and \ref{Tab_Three_Other_Parameters}.
\begin{table}[ht]
\centering
\caption{Parameters pertaining to the three loans, $L_{0}$, $L_{1}$ and $L_{2}$}
\label{Tab_Three_Loan_Parameters}
\begin{tabular}{|c||c|c|c|}
\hline
& Interest & PD & LGD \\
\hline
\hline
$L_{0}$ & $3\%$& $0\%$ & $0\%$ \\
\hline
$L_{1}$ & $9\%$& $6.1\%$ & $10\%$ \\
\hline
$L_{2}$	& $13.2\%$ & $12.2\%$ & $9\%$ \\
\hline
\end{tabular}
\end{table}
\begin{table}[h]
\centering
\caption{The Other Parameters}
\label{Tab_Three_Other_Parameters}
\begin{tabular}{|c|c|}
\hline
$\theta_{1}=1.2\% $ &  $r_{d}=1\%$ \\
\hline
$k_{lev}=4\%$ & $\theta_{2}=1.0\% $ \\
\hline
\end{tabular}
\end{table}
\begin{enumerate}
\item{\textit{Example 1:}} For Example 1, the application of Model \ref{Model3.1} results in the initial investment strategy of \\
$(0\%,59.04\%,40.96\%)$, with the Leverage Ratio being $4\%$. For the intermediate time steps, we discuss the results node-wise.
\begin{enumerate}[(A)]
\item{\textit{Node-3:}} Here we can see that the bank has very limited access to capital via new debt and equity. It can issue a maximum of $2\%$ equity and $1\%$ debt, of the initial total wealth. The Leverage Ratio ($\displaystyle{\frac{Z_{1}-LTD}{Z_{1}}}$) at this point becomes negative. So, the bank will have to raise capital to meet the capital requirement. Mathematically, the bank has to collect $v_{e}$ amount of capital, so that the Leverage Ratio is at least $4\%$, that is,
\[\frac{Z_{1}-\beta_{LT}(1-e)(1+r_{d})}{Z_{1}+v_{e}} \geq 0.04.\]
Substituting all these values, we obtain that $v_{e} \geq 0.0791$. However, the bank is not in a position to issue this much equity, and hence it fails to survive. Therefore bank has to liquidate it's position and goes into bankruptcy.
\item{\textit{Node-2:}} In this scenario, the more risky loan $L_{2}$ has defaulted, but the less risky loan is in good condition. Further, the Leverage Ratio constraint and the ``equity holder's constraint'' are both satisfied without issuing new debt or equity. This leads to a nonempty feasible region. Solving the model, we get the result that, the bank's portfolio is $(35.33\%, 00.05\%)$ (the last one has defaulted) with $v_{d}=0.01$ and $v_{e}=0.012$. Here, due to low raising cost of the equity, the bank raises debt first, and then it raises equity. Theorem \ref{Theorem:3.12} creates an upper bound for the amount of equity that can be raised. Substituting the values, we get the upper bound is $0.012$, which is exactly same as $v_{e}$.
\item{\textit{Node-1:}} In this case all the loans have survived. Therefore bank has higher amount of realization. ``Equity holder's constraint'' protects the bank from failure at the final time point. Therefore, it has a large investment in the safe asset to ensure the return of the equity holders. The portfolio at time $t=1$ is given by $(44.81\%,00.07\%, 20.00\%)$ in the three loans. As a result, the bank raises the available debt and equity ($v_{d}=0.01$ and $v_{e}=0.02$).
\end{enumerate}
\item{\textit{Example 2:}} For Example 2, the application of Model \ref{Model3.1} results in the initial investment strategy of \\ 
$(0\%,54.01\%,45.99\%)$ with the Leverage Ratio being $4\%$. 
\begin{enumerate}[(A)]
\item{\textit{Node-3:}} In this node, the bank has lower short term debt, and has more access to capital via new debt and equity, for which the model has a feasible solution. Hence bank survives this stage, which means that in an unfavorable condition, the bank can raise capital, so that it can survive the unfavorable condition. We solved this model by the differential evolution method and determined that the bank invests $85.39\%$ in the safest asset, and issues $29.94\%$ in equity and $19.94\%$ in debts, that is, it almost raises the amount of capital that it can raise.  From Corollary \ref{Corollary:2}, we see that issuing new equity increases return on equity. And to satisfy capital requirement, the bank has to increase it's Leverage Ratio. From Theorem \ref{Theorem:3.10}, it is evident that bank has to raise equity to increase capital requirement. So bank has to issue new equity. The solution shows that the bank issued that total amount of equity which it can raise, and from Theorem \ref{Theorem:3.11} the upper bound of the debt (depending on this equity) is determined to be $1.64$, which is much higher than the amount that is available. Hence bank is able to issue the full amount of accessible debt. 
\item{\textit{Node-2:}} It also has the same dynamics like the previous problem. In this case, the bank invests $69.87\%$ in the safe asset and $0.04\%$ asset in the less risky loans, with $v_{d}=0.2$, that is, it issues the net amount of debt it can issue. Finally the amount of equity raised is given by, $v_{e}=0.0077$ which is also same as the upper bound given by Theorem \ref{Theorem:3.12}.
\item{\textit{Node-1:}} Here the portfolio is given by $(99.90\%, 5.09\%, 0\%)$. We can see that bank is able to make investments in riskier assets after ensuring the ``equity holder's constraint''. In this case also, the banks issue all the available equity and debt available. Therefore $v_{e}=0.3$ and $v_{d}=0.2$. 
\end{enumerate}
\end{enumerate}
To demonstrate the bank's behavior, we have considered another example, only for \textit{Node-1}. In this example we have changed only $D$ to be $0.9$, with all other parameters being the same as Example 2. It shows that the bank invests its money in safe assets to guarantee the equity holder's expectation of return. Investment in risky assets reduces the return on equity. Therefore to show this phenomenon we have considered this case with $D=0.9$, that is, a bank can issue a large amount of debt. In that case, the result shows that $99.98\%$ is invested in the safe asset (percentage is given with respect to the initial wealth),and $42.61\%$ in the less riskier asset. Compared with the case ($D=0.2$) we can say that banks invest first in safe assets, in order to ensure the constraints and then only it goes for risky assets. Here the ``equity holder's constraint'' shields the bank from default. When a bank satisfies this constraint it cannot fail under any scenario and so there is a large amount of investments in the safe asset. Results show that the return on equity is $2.5674$ for the Example 1 and $1.1209$ for Example 2, while for the last problem it becomes $1.0801$. Therefore, in favorable conditions, the issuance of more equity leads to reduction in its expected return.

\section{Conclusions}
\label{Sec_Three_Conclusions}

The theoretical results of this paper suggest some crucial properties of the bank dynamic problem in the intermediate time point. We have shown that the equity holders' constraint protects banks from defaulting. In favorable scenarios, issuing new equity decreases the return on equity, but in a worst-case scenario, it increases the return on equity. Hence, the bank needs new capital to satisfy the equity holders constraint and the Leverage Ratio constraint. So, if the bank has scope to raise enough capital in a bad condition, it will raise the capital in order to survive. On the other hand, issuing new debt increases return on equity (for lower issuance cost), whereas equity return decreases with issuing new debt (when issuing cost is more considerable). The Leverage Ratio always caps the amount of debt issued. The ``equity holder's constraint'' ensures the return crosses the expected return, if the previous cycle return is higher than the expected return. It desists banks from risk-taking and ensures its survival. Therefore, it reshuffles assets to meet all the constraint requirements.

\section*{Data Availability Statement}

The manuscript has no associated data.

\section*{Conflict of Interest Statement}

The authors declare no conflict of interest.

\bibliographystyle{elsarticle-num}

\bibliography{BIBLIO_3}

\end{document}